\documentclass[a4paper,11pt]{article}
\usepackage{cite}
\usepackage{hyperref}
\usepackage[top=1 in,bottom=1 in,left=1.0 in,right=1.0 in]{geometry}
\hypersetup{hidelinks}
\usepackage{amsthm}
\usepackage{amssymb,amsmath}

\numberwithin{equation}{section}

\newtheorem{lemma}{Lemma}
\newtheorem{proposition}{Proposition}

\newtheorem{theorem}{Theorem}
\newtheorem{remark}{Remark}

\newcommand{\bs}{\boldsymbol}
\newcommand{\mb}{\mathbb}
\newcommand{\mf}{\mathbf}
\newcommand{\mr}{\mathrm}

\newcommand{\bsb}{\begin{subequations}}
\newcommand{\esb}{\end{subequations}}

\newcommand{\ti}[1]{\tilde{#1}}

\usepackage{amsmath}
\usepackage{amsfonts}
\usepackage{latexsym}
\usepackage{amssymb}
\usepackage{appendix}
\usepackage{mathrsfs}
\usepackage{graphicx}
\usepackage{bm}
\usepackage{ulem}
\usepackage{arydshln}
\usepackage{tikz}
\usetikzlibrary{positioning,shapes.misc}
\usepackage{multirow}
\usepackage{rotating}

\usepackage{color}
\newcommand{\bred}{\begin{color}{red}}
\newcommand{\ecl}{\end{color}}
\newcommand{\bblue}{\begin{color}{blue}}
\newcommand{\bgre}{\begin{color}{green}}
\newcommand{\bora}{\begin{color}{orange}}

\begin{document}
\title{\textbf{From the self-dual Yang-Mills equation \\ to the Fokas-Lenells equation}}

\author{Shangshuai Li$^{1,2,3}$
	~~Shuzhi Liu$^4$
	~~ Da-jun Zhang$^{1,2}$\footnote{Corresponding author.
		Email: djzhang@staff.shu.edu.cn}
	~~\\
	{\small $~^1$Department of Mathematics, Shanghai University, Shanghai 200444,    China} \\
	{\small $^{2}$Newtouch Center for Mathematics of Shanghai University,  Shanghai 200444, China}\\
	{\small $^{3}$Department of Applied Mathematics, Faculty of Science and Engineering, Waseda University,}\\
	{\small Tokyo 169-8555, Japan}\\
	{\small $^{4}$School of Statistics and Data Science, Ningbo University of Technology, Ningbo, 315211, China}
		}

\maketitle

\begin{abstract}
A reduction from the self-dual Yang-Mills (SDYM) equation to the unreduced
Fokas-Lenells (FL) system is described in this paper.
It has been known that the SDYM equation can be formulated from the Cauchy matrix schemes
of the matrix Kadomtsev-Petviashvili (KP) hierarchy
and the Ablowitz-Kaup-Newell-Segur (AKNS) hierarchy.
We show that the reduction can be realized in these two Cauchy matrix schemes, respectively.
Each scheme allows us to construct solutions for the unreduced FL system.
We prove that these solutions obtained from different schemes
are equivalent under certain reflection transformation of coordinates.
Using conjugate reduction we obtain solutions of the FL equation.
The paper adds an important example to Ward's conjecture on the reductions of the SDYM equation.
It also indicates the Cauchy matrix structures of the Kaup-Newell hierarchy.

\begin{description}
\item[Keywords:] self-dual Yang-Mills equation, Fokas-Lenells equation,
Cauchy matrix approach, Miura transformation, explicit solution
\end{description}
	
\end{abstract}

\section{Introduction}\label{sec-1}

The self-dual Yang-Mills (SDYM) equation \cite{Yang-1977}, also referred to as the Yang equation
in 4-dimensional Wess-Zumino-Witten model \cite{Donaldson-1985}, 
is an important motion equation in integrable system and twistor theory \cite{Mason-book}.
A  classical  form of the SDYM equation is expressed as
\begin{equation}\label{SDYM}
	\partial_{\ti z}((\partial_zJ)J^{-1})-\partial_{\ti w}((\partial_wJ)J^{-1})=0,
\end{equation}
where $J$ is a matrix function of $(z,\ti z,w,\ti w)\in\mb C^4$.
The SDYM equation allows many reductions towards lower-dimensional classical integrable equations.
For example, it can be reduced to  the Ablowitz-Kaup-Newell-Segur (AKNS) system (see section 2 in \cite{Abiowitz-2003}
for the Korteweg-de Vries (KdV), nonlinear Schr\"{o}dinger (NLS) and sine-Gordon equations),
the Painlev\'{e} equations \cite{Mason-1993}, the chiral field equations \cite{Manakov-1981} and so on.
More than that, there is a well-known conjecture proposed by Ward \cite{Ward-1985}:

\vspace{8pt}
\textit{
	... many (and perhaps all?) of the ordinary and partial differential equations that are regarded as
being integrable or solvable may be obtained from the self-dual gauge field equations (or its generalizations) by reduction.
}
\vspace{8pt}

The purpose of this paper is to build a connection between the SDYM equation and the Fokas-Lenells (FL) equation,
which reads
\begin{align}\label{FL-equation-2}
	u_{xt}-u-2\mathrm{i}|u|^2u_x=0,
\end{align}
where $u=u(x,t)$ is a complex-valued function, $|u|^2=uu^*$, $u^*$ is the conjugate of $u$,
variables $x,t$ are real coordinates,
and $\mathrm{i}$ is the pure imaginary unit.
It bears the names of Fokas and Lenells because Fokas derived it in 1995
from the Hamiltonian triplet of the NLS equation \cite{Fokas-1995}
and later Lenells and Fokas initiated the research of it in a sequence of papers
\cite{Lenells-2009-1,Lenells-2009-2,Lenells-2010}.
They have shown that the equation belongs to the Kaup-Newell hierarchy \cite{Lenells-2009-1}
and it is exactly a reduction of the first member of negative potential Kaup-Newell hierarchy \cite{Lenells-2009-2}:
\bsb\label{FL-unreduce}
\begin{align}
	\label{FL-unreduce-1}&u_{xt}-u-2\mr iuvu_x=0, \\
	\label{FL-unreduce-2}&v_{xt}-v+2\mr ivuv_x=0,
\end{align}
\esb
which we call the pKN$(-1)$ system  for short.
It is worth noting that the FL equation \eqref{FL-equation-2} is equivalent to the
two-dimensional massive Thirring model
and the pKN$(-1)$ system is also known as Mikhailov model.
In 1976 \cite{Mikhailov-1976},  Mikhailov found a Lax pair of the massive Thirring model
which is written in light-cone coordinates as
\begin{align*}
& \mu_x+\mathrm{i}\nu+\mathrm{i} |\nu|^2\mu=0,\\
& \nu_t+\mathrm{i}\mu+\mathrm{i} |\mu|^2 \nu=0,
\end{align*}
and which describes interaction of two states of a fermion.
For more details of the connection between the FL equation and the massive Thirring model,
one can refer to \cite{Kaup-1977,Gerdjikov-1980,Gerdjikov-1982}
or the Appendix A in \cite{Liu-2022}.
As an integrable system, the FL equation has attracted attentions
from various aspects,
for example,  the inverse scattering transform and Riemann-Hilbert method
\cite{Lenells-2009-1,Fan-2021,Ai-2019,Zhang-2023},
bilinear approach \cite{Matsuno-2012,Matsuno-2012-2,Liu-2020,Liu-2022},
Darboux transformation \cite{He-2012,Wang-2020,Ye-2023} and so on.

In this paper, we aim to reduce the SDYM equation to the FL equation.
We will introduce a constraint  such that we can get the pKN$(-1)$ system \eqref{FL-unreduce}
from a general SDYM equation.
Then, we will show that the constraint can be realized based on our recent research
of the SDYM equation using the Cauchy matrix approach.
Recently, we have successfully formulated the SDYM equation \eqref{SDYM}
with two types of Cauchy matrix structures \cite{LSS-2022,LSS-2023}.
The Cauchy matrix approach is a direct method that allows us to
construct integrable equations together with their Lax pairs and solutions with Cauchy matrix structure
through investigating the Sylvester-type equations.
This approach was first established by Nijhoff et al.
to study integrable quadrilateral (discrete) equations \cite{Nijhoff-2009}
in the Adler-Bobenko-Suris list \cite{Adler-2002}.
In their method (see section 2.1 in \cite{Nijhoff-2009}), they introduced a dressed Cauchy  matrix:
\begin{equation}
	\bs M\doteq(M_{j})_{i,j=1,\dots, N},~~~~M_{ij}\doteq \frac{\rho_ic_j}{k_i+k_j},
\end{equation}
where $k_i,c_i$ are constants, $\rho_i=\rho_i(n,m)$ is a discrete plane wave factor  defined as
\begin{align}
	\rho_i\doteq
	\left(\frac{p+k_i}{p-k_i}\right)^n
	\left(\frac{q+k_i}{q-k_i}\right)^m\rho_i^0,~~ (n,m)\in \mathbb{Z}^2,
\end{align}
and $\rho_i^0$ is a phase parameter independent of $(n,m)$.
It turns out that such a matrix $\bs M$ obey
the following Sylvester equation:
\begin{equation}\label{Syl-KdV}
	\bs K\bs M+\bs M\bs K=\bs r\bs c^T,
\end{equation}
where
\begin{equation}
	\bs K\doteq\mr{diag}(k_1,\cdots,k_N),~~~~\bs r\doteq(\rho_1,\cdots,\rho_N)^T,~~~~
	\bs c\doteq(c_1,\cdots,c_N)^T.
\end{equation}
Then a set of scalar functions $S^{(i,j)}\doteq \bs c^T\bs K^j(\bs I_N+\bs M)^{-1}\bs K^i\bs r$
are defined for $i,j\in \mathbb{Z}$,
which turns out to obey some shift and recursive relations.
Here $\bs I_N$ is the $N$-th order identity matrix.
One can also express $S^{(i,j)}$ as the ratio of determinants by
\begin{equation}
	S^{(i,j)}=\frac{g}{f},~~~~f=|\bs I_N+\bs M|,~~~~
	g=-
	\begin{vmatrix}
		\bs I_N+\bs M & \bs K^i\bs r \\
		\bs c^T\bs K^j & 0
	\end{vmatrix}.
\end{equation}
Equations can be derived as closed form of certain $S^{(i,j)}$,
which give rise to discrete integrable equations together with their solutions,
e.g. see \cite{Nijhoff-2009}.
Later, a generalized Cauchy matrix approach was proposed,
which allows $\bs K$ to be arbitrary invertible matrix \cite{Zhang-2013},
providing more flexibility in choosing parameters.
The obtained explicit solutions can be classified in terms of the canonical form of $\bs K$
(or eigenvalue structure of $\bs K$),
which describe interactions of $N$ solitons, resonance of multiple-pole solutions
and interaction between solitons and multiple-pole solutions.
Subsequently, this method was extended to  continuous integrable systems,
enabling the formulation of KdV equation, modified KdV equation, sine-Gordon equation \cite{Xu-2014}, Kadomtsev-Petviashvili (KP) equation \cite{Feng-2022} and $2\times2$ 
Ablowitz-Kaup-Newell-Segur (AKNS) system \cite{Zhao-2018}.
In the recent work \cite{LSS-2022} and \cite{LSS-2023}, we have shown that
the SDYM equation and its solutions can be formulated from the Cauchy matrix structures of the
matrix AKNS system and the matrix KP system.
These progresses will help us realize the reduction constraint (see, e.g. \eqref{decompose-JK})
using the Cauchy matrix structures.  As a result, we will also give a Cauchy-matrix formulation
of the FL equation, which has also emerged in  \cite{Matsuno-2012}.

The paper is organized as follows.
In section \ref{sec-2}, we recall the theory of the SDYM equation
and present dimensional reduction from a general SDYM equation to the pKN$(-1)$ system.
Then, in section \ref{sec-3} we show how the reduction is realized in the Cauchy matrix schemes
of the KP-type and AKNS-type.
Equivalence of these solutions obtained from the two Cauchy matrix schemes
are also discussed in this section.
Section \ref{sec-4} devotes to conjugate reductions such that
$N$-soliton solutions of the FL equation are obtained.
Concluding remarks are given in section \ref{sec-5}.
There are two appendixes. In appendix \ref{App-1}  multiple-pole solutions
in the KP-type Cauchy matrix scheme are constructed.
Appendix \ref{App-2} compares our solutions with Matsuno's solutions
obtained from bilinear approach and demonstrates their uniformity.

\section{From SDYM to pKN$(-1)$ }\label{sec-2}

In this section, first, we  briefly review  the construction towards SDYM equation.
One can refer to \cite{Mason-book,Hamanaka-2020,Huang-2021,LSS-2024}
for more details and descriptions about the theory of SDYM equation.
Then we apply suitable dimensional reduction and coordinate transformation to derive the pKN$(-1)$ system.

\subsection{Theory of the SDYM equation}\label{sec-2-1}

We start from introducing a metric matrix in $\mb C^4=(z_1,z_2,z_3,z_4)$, which is determined by
\begin{equation}\label{eta-mn}
	(\eta^{mn})_{4\times4}\doteq
	\begin{pmatrix}
		0 & 1 & 0 & 0 \\
		1 & 0 & 0 & 0 \\
		0 & 0 & 0 & -1 \\
		0 & 0 & -1 & 0
	\end{pmatrix},~~~~
	m,n=1,2,3,4.
\end{equation}
Let $G$ be a certain Lie group and $g$ be the Lie algebra of $G$.
The Yang-Mills field strengths are determined as follows:
\begin{equation}
	F_{ij}\doteq[\mathcal D_i,\mathcal D_j]=\partial_i A_j-\partial_j A_i+[A_i,A_j],~~~~
	\mathcal D_i\doteq\partial_i+A_i,~~~~i,j\in\{1,2,3,4\},~~~~i\neq j,
\end{equation}
where $[\cdot, ~\cdot]$ denotes the Lie bracket defined as $[A,B]=AB-BA$,
matrix functions $A_i\in g$ are gauge potentials,
operators $\mathcal D_i$ are the covariant derivatives and $\partial_i=\partial_{z_i}$.

The anti-self-dual condition of field strength is given by\footnote{
There is no intrinsic difference between self-dual condition with anti-self-dual condition.
They can be transformed to each other under the coordinate transformation
$(z_1,z_2,z_3,z_4)\rightarrow(z_1,z_2,z_4,z_3)$.}
\begin{equation}\label{ASD-condition}
	F_{ij}=-\frac{1}{2}\epsilon_{ijkl}\eta^{ka}\eta^{lb}F_{ab},
\end{equation}
where $\epsilon_{ijkl}$ is the Levi-Civita tensor, $\eta^{mn}$ follows the definition in \eqref{eta-mn},
and $i,j,k,l,a,b$ are indices running over $\{1,2,3,4\}$, the Einstein summation convention is used.
Denoting $(z,\ti z,w,\ti w)\doteq(z_1,z_2,z_3,z_4)$, one can  rewrite \eqref{ASD-condition} as the follows:
\begin{equation}\label{ASD-condition-expand}
	F_{zw}=0,~~~~F_{\ti z\ti w}=0,~~~~F_{z\ti z}-F_{w\ti w}=0,
\end{equation}
which indicates the existence of $h$ and $\ti h$ satisfying
\begin{equation}
	\mathcal D_zh=0,~~~~\mathcal D_wh=0,~~~~\mathcal D_{\ti z}\ti h=0,~~~~\mathcal D_{\ti w}\ti h=0.
\end{equation}
Defining $J=\ti h^{-1}h$, one can derive the $J$-matrix formulation of the SDYM equation:
\begin{equation}\label{SDYM-J}
	\partial_{\ti z}((\partial_zJ)J^{-1})-\partial_{\ti w}((\partial_wJ)J^{-1})=0.
\end{equation}
The SDYM equation is an integrable system, whose Lax representation is given by
(e.g. \cite{Nimmo-2000})
\bsb\label{SDYM-lax-pair}
\begin{align}
	L(\phi)\doteq (\partial_w-(\partial_wJ)J^{-1})\phi-(\partial_{\ti z}\phi)\zeta=0, \\
	M(\phi)\doteq (\partial_z-(\partial_zJ)J^{-1})\phi-(\partial_{\ti w}\phi)\zeta=0.
\end{align}
\esb
By introducing a Miura transformation
\begin{align}\label{Miura-JK}
	\partial_{\tilde z}K=-(\partial_wJ)J^{-1},~~~~\partial_{\tilde w}K=-(\partial_zJ)J^{-1},
\end{align}
the compatible condition of \eqref{SDYM-lax-pair} also gives rise to the $K$-matrix formulation:
\begin{align}\label{SDYM-K}
	\partial_z\partial_{\ti z}K-\partial_w\partial_{\ti w}K-[\partial_{\ti z}K,\partial_{\ti w}K]=0.
\end{align}

Obviously, the SDYM equation \eqref{SDYM-J} is also a result of the compatibility of \eqref{Miura-JK}.
In principle, the SDYM equation can be studied either in the form \eqref{SDYM-J} or \eqref{SDYM-K},
or in a more general form \eqref{Miura-JK}.
For convenience in the following, we call \eqref{Miura-JK} the general SDYM equation.

\subsection{Dimensional reduction to the pKN$(-1)$ system}\label{sec-2-2}

To reduce the SDYM equation to the pKN$(-1)$ system, let us introduce the following constraints:
\bsb\label{decompose-JK}
\begin{itemize}
\item{choose
\begin{align}
& G=\mr{SL}(2), \label{2.10a}\\
& (z,\ti z,w,\ti w)\in\mb R^4,~~ \mathrm{and}~~   \tilde w=w;
\end{align}
}
\item{assume that $J$ and $K$ have the following variable separation form
\begin{align}
	&J(z,\ti z,w)=\mf e^{-\sigma_3w}J'(z,\ti z)\mf e^{\sigma_3w}, \\
	&K(z,\ti z,w)=\mf e^{-\sigma_3w}K'(z,\ti z)\mf e^{\sigma_3w},
\end{align}
where $\sigma_3=\mr{diag}(1,-1)$ is the third Pauli matrix.
}
\end{itemize}
\esb

Note that the condition \eqref{2.10a} indicates that $|J|$ is a constant
(see, e.g. \cite{Yang-1977,Pohlmeyer-1980} or \cite{LSS-2024}),
therefore we can always normalize it such that
\begin{equation}\label{J=1}
|J|=1.
\end{equation}
The above decomposition for $J$ and $K$ immediately yields
\begin{align}\label{JK-w}
	\partial_wJ=[J,\sigma_3],~~~~\partial_wK=[K,\sigma_3].
\end{align}
Substituting them into the $J$-formulation \eqref{SDYM-J}, the  $K$-formulation \eqref{SDYM-K}
and the general form \eqref{Miura-JK},
we get the following 2-dimensional equations:
\begin{align}
	\label{reduce-J} & \partial_{\ti z}((\partial_zJ)J^{-1})-[[J,\sigma_3]J^{-1},\sigma_3]=0, \\
    \label{reduce-K} & \partial_z\partial_{\ti z}K-[[K,\sigma_3],\sigma_3]-[\partial_{\ti z}K,[K,\sigma_3]]=0,\\
    \label{Miura-pKN} & \partial_zJ=-[K,\sigma_3]J,~~~~\partial_{\ti z}K=-[J,\sigma_3]J^{-1}.
\end{align}
Since $J$ and $K$ are $2\times 2$ matrix functions, we can denote them as
\begin{align}\label{expan-JK}
	J\doteq
	\begin{pmatrix}
		J_{11} & J_{12} \\
		J_{21} & J_{22}
	\end{pmatrix},~~~~
	K\doteq
	\begin{pmatrix}
		K_{11} & K_{12} \\
		K_{21} & K_{22}
	\end{pmatrix}.
\end{align}
Note that the setting $|J|=1$ in \eqref{J=1} indicates
\[J^{-1}=
	\begin{pmatrix}
		J_{22} & -J_{12} \\
		-J_{21} & J_{11}
	\end{pmatrix}.
\]

In the following, we only focus on \eqref{Miura-pKN}, which yields the explicit relations
\bsb\label{Miura-pKN-expand}
\begin{align}\label{Miura-pKN-expand-1}
&	\begin{pmatrix}
		J_{11,z} & J_{12,z} \\
		J_{21,z} & J_{22,z}
	\end{pmatrix}
	=
	\begin{pmatrix}
		2K_{12}J_{21} & 2K_{12}J_{22} \\
		-2K_{21}J_{11} & -2K_{21}J_{12}
	\end{pmatrix}, \\
\label{Miura-pKN-expand-2}
&	\begin{pmatrix}
		K_{11,\ti z} & 	K_{12,\ti z} \\
		K_{21,\ti z} & 	K_{22,\ti z}
	\end{pmatrix}=
	\begin{pmatrix}
		-2J_{12}J_{21} & 2J_{12}J_{11} \\
		-2J_{21}J_{22} & 2J_{21}J_{12}
	\end{pmatrix}.
\end{align}
\esb
Then we give the following theorem.

\begin{theorem}\label{thm-1}
Introduce  new coordinates
\[(x,t)\doteq(2\ti z,2z)\in\mb R^2\]
and functions
	\begin{align}\label{def-uv}
		(u,v)\doteq\left(K_{21}, \mr i\frac{J_{12}}{J_{22}}\right).
	\end{align}
Then $(u,v)$ solves the pKN$(-1)$ system \eqref{FL-unreduce}.
\end{theorem}

\begin{proof}
Using the relations in \eqref{Miura-pKN-expand}, a direct calculation shows that
	\begin{align}
		u_{xt}-u=K_{21,xt}-K_{21}=(J_{11}J_{22}+J_{12}J_{21}-1)K_{21}=2J_{12}J_{21}K_{21}=2\mr ivuu_x,
	\end{align}
where we have made use of the result $|J|=1$.
This is nothing but the first equation in the pKN$(-1)$ system \eqref{FL-unreduce}.
For the second equation \eqref{FL-unreduce-2}, taking $t$-derivative of $v$ yields
\begin{align}\label{thm-1-process}
v_t=\mr i\Bigl(\frac{J_{12}}{J_{22}}\Bigr)_t=\mr i\frac{J_{12,t}J_{22}-J_{12}J_{22,t}}{J_{22}^2}
=\mr iK_{12}-\mr iv^2u.
\end{align}
Further, using \eqref{Miura-pKN-expand}	  we have
	\begin{align}
		v_{xt}=\mr i J_{12}J_{11}-\mr i\frac{J_{12}^2}{J_{22}}J_{21}-2\mr iuvv_x=v-2\mr iuvv_x,
	\end{align}
which is  \eqref{FL-unreduce-2}.
Thus we complete the proof.

\end{proof}

We conclude this subsection with the following remarks.

\begin{remark}\label{rem-1}
The pKN$(-1)$ system \eqref{FL-unreduce} can be reduced from the general SDYM equation \eqref{Miura-JK}
under the reduction constraint \eqref{decompose-JK}.
\end{remark}

\begin{remark}\label{rem-2}
$K_{12}$ and $K_{21}$ in \eqref{reduce-K} enjoy a coupled closed from:
\bsb\label{AKNS-1}
\begin{align}
	\partial_z\partial_{\ti z}K_{12}-4K_{12}+8K_{12}\partial_z^{-1}\partial_{\ti z}(K_{12}K_{21})=0, \\
	\partial_z\partial_{\ti z}K_{21}-4K_{21}+8K_{21}\partial_z^{-1}\partial_{\ti z}(K_{12}K_{21})=0,
\end{align}
\esb
which is known as the first member in the negative AKNS hierarchy (AKNS$(-1)$ system for short)
or the non-potential sine-Gordon system,
(see equation (4.11) in \cite{Zhao-2018} or (2.9) in \cite{Zhang-2009-PD}).
\end{remark}

\begin{remark}\label{rem-3}
The reduction condition \eqref{decompose-JK} is not unique.
 One can  replace $\sigma_3$ in \eqref{decompose-JK} with either of the following,
 \begin{equation}
 \label{P123}
 P_1\doteq\mr{diag}(1,0), ~~ \mathrm{or}~~ P_2\doteq\mr{diag}(0,-1),~~ \mathrm{or}~~
 P_3\doteq\mr{diag}\left(\frac{1}{2},-\frac{1}{2}\right),
 \end{equation}
and introduce $(x,t)\doteq(\ti z,z)\in\mb R^2$.
Then $(u,v)$ defined as in \eqref{def-uv} still solves the pKN$(-1)$ system \eqref{FL-unreduce}.
\end{remark}

\begin{remark}\label{rem-4}
In addition to \eqref{def-uv}, 	 the following setting
	\begin{align}\label{def-uv-2}
		(u,v)\doteq\left(\mr i\frac{J_{21}}{J_{11}},K_{12}\right)
	\end{align}
also satisfies the pKN$(-1)$ system \eqref{FL-unreduce} under the reduction \eqref{decompose-JK}.
\end{remark}

\begin{remark}\label{rem-5}
In practice, if $J$ and $K$ satisfy \eqref{JK-w} together with the setting $|J|=1$,
one can always reduce the general SDYM equation \eqref{Miura-JK} to the pKN$(-1)$ system \eqref{FL-unreduce}.
In the next section, we will show how these conditions are fulfilled in the Cauchy matrix approach.
\end{remark}

\section{Realization of reductions in Cauchy matrix approach}\label{sec-3}

Recalling the Cauchy matrix approach presented in \cite{LSS-2023},
the general SDYM equation \eqref{Miura-JK} can be well defined;
in addition, relations in \eqref{JK-w} and the setting $|J|=1$ can also arise from the Cauchy matrix approach
of the SDYM equation.
In this section, first, we will recall the two Cauchy matrix schemes, namely, the KP-type and the AKNS-type,
in which the general SDYM equation \eqref{Miura-JK} and relation  \eqref{JK-w}
have been established.
This then allows us to realize reductions and present explicit solutions to the pKN$(-1)$ system \eqref{FL-unreduce}.

\subsection{Realization of reductions}\label{sec-3-1}

Recently in \cite{LSS-2023} we have studied the SDYM equation from two types of the Sylvester equations:
\begin{itemize}
	\item The KP-type Sylvester equation (asymmetric Sylvester equation):
	\begin{align}\label{Syl-KP}
		\bs K\bs M-\bs M\bs L=\bs r\bs s^T,
	\end{align}
	where $\bs K,\bs L,\bs M\in\mb C_{N\times N}$, $\bs r=(\bs r_1,\bs r_2)\in\mb C_{N\times 2}$,
$\bs s=(\bs s_1,\bs s_2)\in\mb C_{N\times 2}$.
	By introducing $\bs M_1,\bs M_2$ that satisfy
	\begin{align}
		\bs K\bs M_1-\bs M_1\bs L=\bs r_1\bs s_1^T, ~~~~
		\bs K\bs M_2-\bs M_2\bs L=\bs r_2\bs s_2^T,
	\end{align}
	we have $\bs M=\bs M_1+\bs M_2$.
	The master functions are defined as
	\begin{align}
		\bs S^{(i,j)}_{[\mr{KP}]}&=\bs s^T\bs L^j \bs M_1^{-1}\bs K^i\bs r
       =\bs s^T\bs L^j(\bs M_1+\bs M_2)^{-1}\bs K^i\bs r
		=\begin{pmatrix}
			s_{11}^{(i,j)} & s_{12}^{(i,j)} \\
			s_{21}^{(i,j)} & s_{22}^{(i,j)}
		\end{pmatrix} \notag \\
		&=
		\begin{pmatrix}
			\bs s_1^T\bs L^j(\bs M_1+\bs M_2)^{-1}\bs K^i\bs r_1 &
\bs s_1^T\bs L^j(\bs M_1+\bs M_2)^{-1} \bs K^i\bs r_2 \\
			\bs s_2^T\bs L^j(\bs M_1+\bs M_2)^{-1}\bs K^i\bs r_1 &
\bs s_2^T\bs L^j(\bs M_1+\bs M_2)^{-1} \bs K^i\bs r_2 \\
		\end{pmatrix}.\label{3.3}
	\end{align}
	
	\item The AKNS-type Sylvester equation (symmetric Sylvester equation):
	\begin{align}\label{Syl-AKNS}
		\bs K\bs M-\bs M\bs K=\bs r\bs s^T,
	\end{align}
	where $\bs K,\bs M,\bs r,\bs s$ are block matrices in the forms of
	\begin{align}\label{3.5}
		\bs K=
		\begin{pmatrix}
			\bs K_1 & \\
			& \bs K_2
		\end{pmatrix},~~~~
		\bs M=
		\begin{pmatrix}
			& \bs M_1 \\
			\bs M_2 &
		\end{pmatrix},~~~~
		\bs r=
		\begin{pmatrix}
			\bs r_1 & \\
			& \bs r_2
		\end{pmatrix},~~~~
		\bs s=
		\begin{pmatrix}
			& \bs s_1 \\
			\bs s_2 &
		\end{pmatrix},
	\end{align}
	with $\bs K_i,\bs M_i\in\mb C_{N\times N}$, $\bs r_i,\bs s_i\in\mb C_{N\times1}$, $i=1,2$.
Expansion of \eqref{Syl-AKNS} yields
	\begin{align}
		\bs K_1\bs M_1-\bs M_1\bs K_2=\bs r_1\bs s_2^T, ~~~~
		\bs K_2\bs M_2-\bs M_2\bs K_2=\bs r_2\bs s_1^T.
	\end{align}
	The master functions are defined as
	\begin{align}\label{3.7}
		\bs S^{(i,j)}_{[\mr{AKNS}]}&=\bs s^T\bs K^j(\bs I_{2N}+\bs M)^{-1}\bs K^i\bs r
		=\begin{pmatrix}
			s_1^{(i,j)} & s_2^{(i,j)} \\
			s_3^{(i,j)} & s_4^{(i,j)}
		\end{pmatrix}\notag  \\
		&=
		\begin{pmatrix}
			\bs s_2^T\bs K_2^j(\bs M_1-\bs M_2^{-1})^{-1}\bs K_1^i\bs r_1 &
\bs s_2^T\bs K_2^j(\bs I_N-\bs M_2\bs M_1)^{-1}\bs K_2^i\bs r_2 \\
			\bs s_1^T\bs K_1^j(\bs I_N-\bs M_1\bs M_2)^{-1}\bs K_1^i\bs r_1 &
\bs s_1^T\bs K_1^j(\bs M_2-\bs M_1^{-1})^{-1}\bs K_2^i\bs r_2
		\end{pmatrix}.
	\end{align}
\end{itemize}

They give rise to two different formulations of solutions of the SDYM equation  \eqref{SDYM}. For convenience,
we call the solutions derived from \eqref{Syl-KP}/\eqref{Syl-AKNS} the KP/AKNS-type solution, respectively.

For the case of the KP-type, we introduce the following dispersion relations:
\begin{align}\label{dispersion-SDYM}
	\bs r_{x_n}=\bs K^n\bs r\bs a,~~~~\bs s_{x_n}=-(\bs L^T)^n\bs s\bs a,~~~~n\in\mb Z,
\end{align}
where $\bs a=\mr{diag}(a_1,a_2)$. Then the setting
\begin{align}\label{def-UV}
	(\bs V,\bs U)\doteq(\bs I_2-\bs S^{(-1,0)}_{[\mr{KP}]},\bs S^{(0,0)}_{[\mr{KP}]})
\end{align}
yields a differential recurrence relation (see Appendix A in \cite{LSS-2023}):
\begin{align}\label{VVU}
	\bs V_{x_{n+1}}\bs V^{-1}=-\bs U_{x_n}.
\end{align}
This actually provides the general SDYM equation \eqref{Miura-JK} through taking different $n$.
In addition, for both $\bs U$ and $\bs V$, there hold the relations (see equation (2.5b) in \cite{LSS-2023})
\begin{equation}\label{VU-a}
	\bs U_{x_{0}}=[\bs U, \bs a], ~~~ \bs V_{x_{0}}=[\bs V, \bs a].
\end{equation}
Thus, in summary, if we take $\tilde w=w=x_0$ in the above relations and $n=0,-1$ in \eqref{VVU},
namely,
\begin{align}\label{VVU-2}
	\bs V_{x_{1}}\bs V^{-1}=-\bs U_{x_0}, ~~~
\bs V_{x_{0}}\bs V^{-1}=-\bs U_{x_{-1}},
\end{align}
we can recover \eqref{JK-w} and \eqref{Miura-JK} by choosing
\begin{equation}\label{3.13}
(\bs U, \bs V, x_{-1}, x_1, x_0)=(K, J, \tilde z, z, \tilde w=w).
\end{equation}
Note also that in this case $|\bs V|=|\bs L|/|\bs K|$ (see Theorem 2 in \cite{LSS-2023}),
which can be normalized to be $1$.
Thus, the reduction from the general SDYM equation \eqref{Miura-JK} to the pKN$(-1)$ system \eqref{FL-unreduce}
can be realized in the Cauchy matrix scheme of the KP-type.

For the case of the AKNS-type, the dispersions are introduced by replacing $\bs L$ in \eqref{dispersion-SDYM}
with $\bs K$, i.e.
\begin{align}\label{dispersion-AKNS}
	\bs r_{x_n}=\bs K^n\bs r\bs a,~~~~\bs s_{x_n}=-(\bs K^T)^n\bs s\bs a,~~~~n\in\mb Z,
\end{align}
and we introduce two functions as
\begin{align}\label{3.15}
	(\bs v,\bs u)\doteq\left(\bs I_2-\bs S^{(-1,0)}_{[\mr{AKNS}]},\bs S^{(0,0)}_{[\mr{AKNS}]}\right).
\end{align}
There are similar relations (see Theorem 1 in \cite{LSS-2022}):
\begin{align}
	\bs v_{x_{n+1}}\bs v^{-1}=-\bs u_{x_n}, ~~~n\in\mb Z,
\end{align}
and (see equation (2.5b) in \cite{LSS-2023})
\begin{equation}\label{VU-a}
	\bs u_{x_{0}}=[\bs u, \bs a], ~~~ \bs v_{x_{0}}=[\bs v, \bs a].
\end{equation}
Then, similarly, one can  recover \eqref{JK-w} and \eqref{Miura-JK} by choosing
\begin{equation}\label{3.18}
(\bs u, \bs v, x_{-1}, x_1, x_0)=(K, J, \tilde z, z, \tilde w=w).
\end{equation}
We can also have $|\bs v|=1$ in this case (see Theorem 2 in \cite{LSS-2023}).
Thus, the reductions are realized from the Cauchy matrix scheme of the AKNS-type as well.

Now that these reductions can be realized from the Cauchy matrix approach,
we can derive explicit solutions for the pKN$(-1)$ system by using the known results
of the (matrix) KP hierarchy and AKNS hierarchy.

\subsection{Solutions of pKN$(-1)$ system from Cauchy matrix scheme}\label{sec-3-2}

\subsubsection{From KP-type Cauchy matrix scheme}\label{sec-3-2-1}

We have shown that the reduction conditions can be fulfilled in the two Cauchy matrix schemes.
That means we can then get explicit solution $(u,v)$ for the pKN$(-1)$ system through
the formulation \eqref{def-uv}.
In practice, we need to in the first step present explicit expression
for $\bs U$ and $\bs V$ defined in \eqref{def-UV}
and then recover $(u,v)$ from  \eqref{def-uv}.
To achieve that, we consider the following solutions of the Sylvester equation \eqref{Syl-KP} together with
the dispersion relation \eqref{dispersion-SDYM} (see \cite{LSS-2023}):
\bsb
\begin{align}
&\bs K=\mr{diag}(k_1,\cdots,k_N), ~~~~\bs L=\mr{diag}(l_1,\cdots,l_N), ~~~~ k_i, l_i\in \mathbb{C},\\
\label{r_j}& \bs r=(\bs r_1,\bs r_2), ~~~~ \bs r_j=(\rho_j(k_1),\cdots,\rho_j(k_N))^T,\\
\label{s_j}& \bs s=(\bs s_1,\bs s_2), ~~~~ \bs s_j=(\sigma_j(l_1),\cdots,\sigma_j(l_N))^T,~~~~j=1,2, \\
\label{M-Cauchy-matrix}&\bs M=\bs M_1+\bs M_2,~~~~
\bs M_j=(M_{is}^{(j)})_{N\times  N},~~~~
M_{is}^{(j)}=\frac{\rho_j(k_i)\sigma_j(l_s)}{k_i-l_s},~~~~ j=1,2,
\end{align}
\esb
with the plane wave factors
\begin{align}\label{PWF-FL}
\rho_j(k_i)=\exp\left(a_j\mathcal{L}(k_i)+\lambda_j(k_i)\right),~~
\sigma_j(l_i)=\exp\left(-a_j\mathcal{L}(l_i)+\mu_1(l_i)\right),~~~~ j=1,2,
\end{align}
where
\begin{align}\label{L}
		\mathcal{L}(k)\doteq k^nx_n+k^{n+1}x_{n+1}+k^mx_m+k^{m+1}x_{m+1},
\end{align}
and the phase factors $\lambda_j(k), \mu_j(k)$ are functions of $k$.

In light of \eqref{3.3}, \eqref{def-UV} and \eqref{3.13}, we have
\begin{equation}
		\bs V\doteq \bs I_2-\bs s^T\bs M^{-1}\bs K^{-1}\bs r,~~~~
		\bs U\doteq\bs s^T\bs M^{-1}\bs r,
\end{equation}
\begin{equation}
J=	\bs V=
	\begin{pmatrix}
		V_{11} & V_{12} \\
		V_{21} & V_{22}
	\end{pmatrix},~~~~
K=	\bs U=
	\begin{pmatrix}
		U_{11} & U_{12} \\
		U_{21} & U_{22}
	\end{pmatrix},
\end{equation}
and further, from \eqref{def-uv} we have
\begin{align}\label{def-uv-asym}
(u,v)=\left(U_{21},\mr i\frac{V_{12}}{V_{22}}\right)
=\left(\bs s^T_2\bs M^{-1}\bs r_1,-\frac{\mr i\bs s_1^T\bs M^{-1}\bs K^{-1}\bs r_2}
{1-\bs s_2^T\bs M^{-1}\bs K^{-1}\bs r_2}\right),
\end{align}
which will be solutions of the pKN$(-1)$ system \eqref{FL-unreduce} if taking
$(n,m)=(-1,0)$ in \eqref{L} and
\begin{equation}\label{a-xt}
\bs a=\sigma_3, ~~~ (x,t)\doteq(2x_{-1},2x_1).
\end{equation}
Note that $|\bs V|$ is a constant given by (see Theorem 2 in \cite{LSS-2023})
$|\bs V|=|\bs L|/|\bs K|$
and gives rise to same $(u,v)$ through \eqref{def-uv-asym}
no matter $|\bs V|$ is normalized to be $1$ or not.
So, in the following, no need to normalize $\bs V$.
In addition, corresponding to Remark \ref{rem-3} we have the following.

\begin{remark}\label{rem-6}
In the case $\bs a$ takes either of $P_j$ for $j=1,2,3$
as given in \eqref{P123}, we shall take
$(x,t)\doteq(x_{-1},x_1)$ and take the plane wave factors as the following,
\begin{align}\label{3.26}
	\rho_j(k_i)=\exp\left(a_j\left(\frac{1}{k_i}x+k_it\right)+\lambda'_j(k_i)\right), ~~
 \sigma_j(l_i)=\exp\left(-a_j\left(\frac{1}{l_i}x+l_it\right)+\mu'_j(l_i)\right),
\end{align}
for $j=1,2$, where the phase factors have been taken as
\begin{align}
\lambda'_j(k_i)=\lambda_j(k_i)+2a_jx_0, ~~
\mu'_j(k_i)=\mu_j(l_j)-2a_jx_0, ~~~j=1,2.
\end{align}
 \end{remark}

Next, we present the $(u,v)$ formulation in a more explicit form.
Note that the dressed Cauchy matrix $\bs M$ can be expressed in the following decomposition form:
\begin{align}
	\bs M=\bs M_1+\bs M_2,~~~~
	\bs M_1=\bs R_1\bs G\bs S^T_1,~~~~\bs M_2=\bs R_2\bs G\bs S^T_2,
\end{align}
where
\bsb
\begin{align}
	&\bs R_j=\mr{diag}(\rho_j(k_1),\cdots,\rho_j(k_N)),~~~~
	\bs S_j=\mr{diag}(\sigma_j(l_1),\cdots,\sigma_j(l_N)),~~~~j=1,2, \\
	\label{Cauchy-G}&\bs G=(G_{ij})_{N\times N}=\frac{1}{k_i-l_j},~~~~i,j=1,\cdots, N.
\end{align}
\esb
This fact indicates that we can rewrite $\bs r_j$ and $\bs s_j$ as
\begin{align}
	\bs r_j=\bs R_j\bs e_N,~~~~\bs s_j=\bs S_j\bs e_N,~~~~\bs e_N=(\underbrace{1,1,\cdots,1}_{\text{$N$-dimensional}})^T,~~~~j=1,2.
\end{align}
Thus for each $s_{ab}^{(i,j)}$ in $\bs S^{(i,j)}_{[\mr{KP}]}$, $a,b=1,2$, we have
\bsb
\begin{align}
	&s_{11}^{(i,j)}=\bs e_N^T\bs L^j(\bs G+\bs R_1^{-1}\bs R_2\bs G\bs S_2^T(\bs S_1^T)^{-1})^{-1}
\bs K^i\bs e_N, \\
	&s_{12}^{(i,j)}=\bs e^T_N\bs L^j(\bs R_2^{-1}\bs R_1\bs G+\bs G\bs S_2^T(\bs S_1^T)^{-1})^{-1}
\bs K^i\bs e_N, \\
	&s_{21}^{(i,j)}=\bs e_N^T\bs L^j(\bs G\bs S_1^T(\bs S_2^T)^{-1}+\bs R_1^{-1}\bs R_2\bs G)^{-1}
\bs K^i\bs e_N, \\
	&s_{22}^{(i,j)}=\bs e_N^T\bs L^j(\bs R_2^{-1}\bs R_1\bs G\bs S_1^T(\bs S_2^T)^{-1}+\bs G)^{-1}
\bs K^i\bs e_N.
\end{align}
\esb
Thus we can rewrite the formula \eqref{def-uv-asym} as
\bsb
\begin{align}
	&u=\bs s_2^T\bs M^{-1}\bs r_1=
	\bs e_N^T(\bs G\bs S_1^T(\bs S_2^T)^{-1}
+\bs R_1^{-1}\bs R_2\bs G)^{-1}\bs e_N, \\
	&v=-\frac{\mr i\bs s_1^T\bs M^{-1}\bs K^{-1}\bs r_2}{1-\bs s_2^T\bs M^{-1}\bs K^{-1}\bs r_2}
	=
	-\frac{\mr i\bs e_N^T(\bs R_2^{-1}\bs R_1\bs G+\bs G\bs S_2^T(\bs S_1^T)^{-1})^{-1}
\bs K^{-1}\bs e_N}{1-\bs e_N^T(\bs R_2^{-1}\bs R_1\bs G\bs S_1^T(\bs S_2^T)^{-1}+\bs G)^{-1}
\bs K^{-1}\bs e_N}.
\end{align}
\esb
Notice that in this explicit formulation, we can introduce
\bsb
\begin{align}
	&\bs R\doteq\bs R_1^{-1}\bs R_2=\mr{diag}\left(\frac{\rho_2(k_1)}{\rho_1(k_1)},\cdots,\frac{\rho_2(k_N)}{\rho_1(k_N)}\right), \\
	&\bs S\doteq\bs S_1^{-1}\bs S_2=\mr{diag}\left(\frac{\sigma_2(l_1)}{\sigma_1(l_1)},\cdots,\frac{\sigma_2(l_N)}{\sigma_1(l_N)}\right),
\end{align}
\esb
where
\bsb
\begin{align}
	&\frac{\rho_2(k_i)}{\rho_1(k_i)}
=\exp\left((a_2-a_1)\left(\frac{1}{k_i}x+k_it\right)+\lambda'_2(k_i)-\lambda'_1(k_i)\right), \\
	&\frac{\sigma_2(l_i)}{\sigma_1(l_i)}
=\exp\left((a_1-a_2)\left(\frac{1}{l_i}x+l_it\right)+\mu'_2(l_i)-\mu'_1(l_i)\right).
\end{align}
\esb
This fact also indicates that different $\bs a$ may lead to the same results.
Finally, we have explicit formulas for  $(u,v)$ as
\bsb\label{uv-general-KP}
\begin{align}
	&u=\bs e_N^T(\bs G\bs S^{-1}+\bs R\bs G)^{-1}\bs e_N,\\
	&v=-\frac{\mr i\bs e^T_N(\bs R^{-1}\bs G+\bs G\bs S)^{-1}\bs K^{-1}\bs e_N}{1-\bs e_N^T
(\bs R^{-1}\bs G\bs S^{-1}+\bs G)^{-1}\bs K^{-1}\bs e_N},
\end{align}
\esb
where we have taken $\bs a$ to be either of $P_j$ as given in \eqref{P123}, and
\bsb\label{def-RS}
\begin{align}
	&\bs R=\mr{diag}(r(k_1),\cdots,r(k_N)),~~~~
	r(k_i)=\exp\left(-\frac{1}{k_i}x-k_it+\zeta(k_i)\right), \\
	&\bs S=\mr{diag}(s(l_1),\cdots,s(l_N)),~~~~
	s(l_i)=\exp\left(\frac{1}{l_i}x+l_it+\eta(l_i)\right)
\end{align}
with $\zeta(k_i)$ and $\eta(l_i)$ being phase factors.
\esb

\subsubsection{From AKNS-type Cauchy matrix scheme}\label{sec-3-2-2}

In the case of the AKNS-type, the Sylvester equation \eqref{Syl-AKNS} with \eqref{3.5}
and the dispersion relation \eqref{dispersion-AKNS} have solutions
(see \cite{Zhao-2018,LSS-2022,LSS-2023}):
\bsb\label{def-sym}
	\begin{align}
		&\bs K_1=\mr{diag}(k_1,\cdots,k_N),~~~~\bs K_2=\mr{diag}(l_1,\cdots,l_N), \\
		&\bs r_1=(\rho_1(k_1),\cdots,\rho_1(k_N))^T,~~~~\bs s_1=(\sigma_1(k_1),\cdots,\sigma_1(k_N))^T,\\
		&\bs r_2=(\rho_2(l_1),\cdots,\rho_2(l_N))^T,~~~~\bs s_2=(\sigma_2(l_1),\cdots,\sigma_2(l_N))^T,\\
		\label{M1}&\bs M_1=(M_{1,ij})_{N\times N},~~~~M_{1,ij}=\frac{\rho_1(k_i)\sigma_2(l_j)}{k_i-l_j}, \\
		\label{M2}&\bs M_2=(M_{2,ij})_{N\times N},~~~~M_{2,ij}=\frac{\rho_2(l_i)\sigma_1(k_j)}{l_i-k_j},
	\end{align}
	\esb
where the plane wave factors $\rho_j, \sigma_j$ are defined as in \eqref{PWF-FL}.

In light of \eqref{3.7}, \eqref{3.15} and \eqref{3.18}, we have
\begin{align}
		\bs v\doteq \bs I_2-\bs s^T(\bs I+\bs M)^{-1}\bs K^{-1}\bs r,~~~~
		\bs u\doteq \bs s^T(\bs I+\bs M)^{-1}\bs r.
\end{align}
Suppose
\begin{align}
J=\bs v=
	\begin{pmatrix}
		v_{1} & v_{2} \\
		v_{3} & v_{4}
	\end{pmatrix},~~~~
K=\bs u=
	\begin{pmatrix}
		u_{1} & u_{2} \\
		u_{3} & u_{4}
	\end{pmatrix}.
\end{align}
Then, the following functions
\begin{align}\label{def-uv-sym}
		(u,v)\doteq \left(u_{3},\mr i\frac{v_{2}}{v_{4}}\right)=\left(\bs s_1^T(\bs I_N-\bs M_1\bs M_2)^{-1}
\bs r_1,\frac{-\mr i\bs s_2^T(\bs I_N-\bs M_2\bs M_1)^{-1}\bs K_2^{-1}\bs r_2}{1+\bs s_1^T
(\bs I_N-\bs M_1\bs M_2)^{-1}\bs M_1\bs K_2^{-1}\bs r_2}\right)
	\end{align}
provide solutions to the  pKN$(-1)$ system \eqref{FL-unreduce} where we take
$(n,m)=(-1,0)$ in \eqref{L} and take \eqref{a-xt} as well.
Note that in the case of the AKNS-type, $|J|=1$ (see Theorem 2 in \cite{LSS-2023}
and the Remark \ref{rem-6} holds too
and thus in the following we use the plane wave factors defined in the form \eqref{3.26}.

To give explicit expressions of $(u,v)$,
we factorize $\bs M_1$ and $\bs M_2$ as the following:
\begin{align}
	\bs M_1=\bs R_1\bs G\bs S_2^T,~~~~\bs M_2=-\bs R_2\bs G^T\bs S_1^T,
\end{align}
where
\bsb
\begin{align}
	&\bs R_1=\mr{diag}(\rho_1(k_1),\cdots,\rho_1(k_N)),~~~~
	\bs S_2=\mr{diag}(\sigma_2(l_1),\cdots,\sigma_2(l_N)), \\
	&\bs R_2=\mr{diag}(\rho_2(l_1),\cdots,\rho_2(l_N)),~~~~
	\bs S_1=\mr{diag}(\sigma_1(k_1),\cdots,\sigma_1(k_N)), \\
	&\bs G=\left(\frac{1}{k_i-l_j}\right)_{1\leq i,j\leq N},~~~~
\bs G^T=\left(-\frac{1}{l_i-k_j}\right)_{1\leq i,j\leq N},~~~~
	\bs e_N=(\underbrace{1,1,\cdots,1}_{\text{$N$-dimensional}}),
\end{align}
\esb
and the plane wave factors are defined as in \eqref{3.26}.
For each entry in $\bs S^{(i,j)}_{[\mr{AKNS}]}$, we can rewrite them as
\bsb
\begin{align}
	&s_1^{(i,j)}=\bs e_N^T\bs K_2^j(\bs G+\bs R_1^{-1}(\bs S_1^T)^{-1}(\bs G^T)^{-1}\bs R_2^{-1}
(\bs S_2^T)^{-1})^{-1}\bs K_1^i\bs e_N, \\
	&s_2^{(i,j)}=\bs e_N^T\bs K_2^j(\bs R_2^{-1}(\bs S_2^T)^{-1}+\bs G^T\bs S_1^T\bs R_1\bs G)^{-1}
\bs K_2^i\bs e_N, \\
	&s_3^{(i,j)}=\bs e_N^T\bs K_1^j(\bs R_1^{-1}(\bs S_1^T)^{-1}+\bs G\bs S_2^T\bs R_2\bs G^T)\bs K_1^i
\bs e_N, \\
	&s_4^{(i,j)}=-\bs e_N^T\bs K_1^j(\bs G^T+\bs R_2^{-1}(\bs S_2^T)^{-1}\bs G^{-1}\bs R_1^{-1}
(\bs S_1^T)^{-1})^{-1}\bs K_2^i\bs e_N.
\end{align}
\esb
Then we  introduce
\bsb
\begin{align}
	&\bs P\doteq\bs R_1\bs S_1^T=\mr{diag}(\rho_1(k_1)\sigma_1(k_1),\cdots,\rho_1(k_N)\sigma_1(k_N)), \\
	&\bs Q\doteq\bs R_2\bs S_2^T=\mr{diag}(\rho_2(l_1)\sigma_2(l_1),\cdots,\rho_2(l_N)\sigma_2(l_N)).
\end{align}
\esb
Finally,  $u$ and $v$ can be expressed as
\bsb\label{uv-general-AKNS}
\begin{align}
	\label{uv-general-AKNS-1}&u=\bs e_N^T(\bs P^{-1}+\bs G\bs Q\bs G^T)^{-1}\bs e_N, \\
	\label{uv-general-AKNS-2}&v=\frac{-\mr i\bs e_N^T(\bs Q^{-1}+\bs G^T\bs P\bs G)^{-1}\bs K_2^{-1}
\bs e_N}{1+\bs e^T_N(\bs G^T+\bs Q^{-1}\bs G^{-1}\bs P^{-1})^{-1}\bs K_2^{-1}\bs e_N},
\end{align}
\esb
where we have taken $\bs a$ to be either of $P_j$ as given in \eqref{P123}, and
\bsb\label{def-PQ}
\begin{align}
	&\bs P=\mr{diag}(p(k_1),\cdots,p(k_N)),~~~~
	p(k_i)=\exp\left(\frac{1}{k_i}x+k_it+\omega(k_i)\right), \\
	&\bs Q=\mr{diag}(q(l_1),\cdots,q(l_N)),~~~~
	q(l_i)=\exp\left(-\frac{1}{l_i}x-l_it+\theta(l_i)\right).
\end{align}
with $\omega(k_i)$ and $\theta(l_i)$ being phase factors.
\esb

The following result follows from Remark \ref{rem-4}.

\begin{remark}\label{rem-7}
The pKN$(-1)$ system \eqref{FL-unreduce} also has solutions:
	\bsb\label{uv-general-AKNS-new}
	\begin{align}
		\label{uv-general-AKNS-new-1}&u=\mr i\frac{-s_3^{(-1,0)}}{1-s_1^{(-1,0)}}
=\frac{-\mr i\bs e_N^T(\bs P^{-1}+\bs G\bs Q\bs G^T)^{-1}\bs K_1^{-1}\bs e_N}
{1-\bs e_N^T(\bs G+\bs P^{-1}(\bs G^T)^{-1}\bs Q^{-1})^{-1}\bs K^{-1}\bs e_N}, \\
		\label{uv-general-AKNS-new-2}&v=s_2^{(0,0)}
=\bs e_N^T(\bs Q^{-1}+\bs G^T\bs P\bs G)^{-1}\bs e_N,
	\end{align}
	\esb
where $\bs P,\bs Q$ are defined in \eqref{def-PQ}.
\end{remark}

\subsubsection{Equivalence of the two types of solutions}\label{sec-3-2-3}

We can prove that the above two types of solutions are equivalent in some sense.

For Cauchy matrix $\bs G$ defined in \eqref{Cauchy-G}, its inverse can be represented in terms of $\bs G$ as
\begin{align}
		\bs G^{-1}=\bs X\bs G^T\bs Y,
\end{align}
where
\bsb
\begin{align}
&\bs X=\mr{diag}(X_1,\cdots,X_N), ~~~X_i=\frac{\Pi_{s=1}^N(k_s-l_i)}{\Pi_{1\leq s\leq N}^{s\neq i}(l_s-l_i)}, \\
&\bs Y=\mr{diag}(Y_1,\cdots,Y_N), ~~~Y_i=\frac{\Pi_{s=1}^N(k_i-l_s)}{\Pi_{1\leq s\leq N}^{s\neq i}(k_i-k_s)}.
\end{align}
\esb
In the following, we recover the KP-type solution $(u_{[\mr{KP}]},v_{[\mr{KP}]})$ given in \eqref{uv-general-KP}
from the AKNS-type solution $(u_{[\mr{AKNS}]},v_{[\mr{AKNS}]})$ given in \eqref{uv-general-AKNS-new}.

Starting from \eqref{uv-general-AKNS-new-2}, we can rewrite $v_{[\mr{AKNS}]}$ as:
\begin{align*}
	v_{[\mr{AKNS}]}&=\bs e_N^T(\bs Q^{-1}+\bs G^T\bs P\bs G)^{-1}\bs e_N \\
	&=\bs e_N^T(\bs Q^{-1}+\bs X^{-1}\bs G^{-1}\bs Y^{-1}\bs P\bs G)^{-1}\bs e_N \\
	&=\bs e_N^T(\bs G\bs X\bs Q^{-1}+\bs Y^{-1}\bs P\bs G)^{-1}\bs e_N,
\end{align*}
where $\bs P,\bs Q$ are defined in \eqref{def-PQ}
and for $\bs G\bs X\bs e_N$, we have used a property in Lagrange polynomial (see Lemma 2.4 in \cite{Alison-2015}):
\begin{align}
	\bs G\bs X\bs e_N=\left(\sum_{i=1}^N\frac{\Pi_{s\neq1}(k_s-l_i)}{\Pi_{s\neq i}(l_s-l_i)},\cdots,\sum_{i=1}^N\frac{\Pi_{s\neq N}(k_s-l_i)}{\Pi_{s\neq i}(l_s-l_i)}\right)^T=(1,\cdots,1)^T=\bs e_N.
\end{align}
Then, letting $\theta(l_i)=\eta(l_i)+\ln(X_i)$ and $\omega(k_i)=\zeta(k_i)+\ln(Y_i)$,   we have
\begin{align}
	\bs X\bs Q^{-1}(-x,-t)=\bs S^{-1}(x,t),~~~~\bs Y^{-1}\bs P(-x,-t)=\bs R(x,t).
\end{align}
Thus
\begin{align}
	v_{[\mr{AKNS}]}(-x,-t)&=\bs e_N^T(\bs G\bs X\bs Q^{-1}(-x,-t)+\bs Y^{-1}\bs P(-x,-t)\bs G)^{-1}\bs e_N
\nonumber\\
	&=\bs e_N^T(\bs G\bs S^{-1}(x,t)+\bs R(x,t)\bs G)^{-1}\bs e_N=u_{[\mr{KP}]}(x,t).\label{3.53}
\end{align}
Applying the same trick leads to a similar result for $u_{[\mr{AKNS}]}$:
\begin{align*}
	u_{[\mr{AKNS}]}&=
	\frac{-\mr i\bs e_N^T(\bs P^{-1}+\bs G\bs Q\bs G^T)^{-1}\bs K_1^{-1}\bs e_N}{1-\bs e_N^T
(\bs G+\bs P^{-1}(\bs G^T)^{-1}\bs Q^{-1})^{-1}\bs K_1^{-1}\bs e_N}  \\
	&=\frac{-\mr i\bs e_N^T(\bs P^{-1}\bs Y\bs G+\bs G\bs Q\bs X^{-1})^{-1}\bs K_1^{-1}\bs e_N}
{1-\bs e_N^T(\bs P^{-1}\bs Y\bs G\bs X\bs Q^{-1}+\bs G)^{-1}\bs K_1^{-1}\bs e_N}.
\end{align*}
Thus we find
\begin{align}
	u_{[\mr{AKNS}]}(-x,-t)=\frac{-\mr i\bs e_N^T(\bs R^{-1}(x,t)\bs G+\bs G\bs S(x,t))^{-1}\bs K^{-1}
\bs e_N}{1-\bs e_N^T(\bs R^{-1}(x,t)\bs G\bs S^{-1}(x,t)+\bs G)^{-1}\bs K^{-1}\bs e_N}
=v_{[\mr{KP}]}(x,t). \label{3.54}
\end{align}

\begin{proposition}\label{prop-1}
The solutions derived from the Cauchy matrix scheme of the KP-type in Sec.\ref{sec-3-2-1}
and derived from the scheme of the AKNS-type in Sec.\ref{sec-3-2-2}
are connected as
\begin{equation}
(u_{[\mr{KP}]}(x,t), v_{[\mr{KP}]}(x,t))=
(v_{[\mr{AKNS}]}(-x,-t), u_{[\mr{AKNS}]}(-x,-t)),
\end{equation}
which  coincides with the fact that
if $(u(x,t),v(x,t))$ solves the pKN$(-1)$ system \eqref{FL-unreduce}, so does  $(v(-x,-t),u(-x,-t))$.
\end{proposition}

\section{Conjugate reduction to the FL equation}\label{sec-4}

The FL equation \eqref{FL-equation-2} is a result of the conjugate reduction of the pKN$(-1)$ system
\eqref{FL-unreduce} by taking $v=u^*$.
Such a reduction can be realized by imposing constraints on $\bs K$ and $\bs L$
(or $\bs K_1$ or $\bs K_2$ for the AKNS-type)
and generates solutions for the FL equation from those solutions of the  pKN$(-1)$ system
that we got in Sec.\ref{sec-3-2}.
In this section, we implement such reductions so that we can obtain solution for the FL equation \eqref{FL-equation-2}.

\subsection{The KP-type solution}\label{sec-4-1}

For those solutions obtained in Sec.\ref{sec-3-2-1} from the KP-type formulation,
we describe the reduction and constraint in the following theorem.

\begin{theorem}\label{thm-2}
	For $(u,v)$ given in \eqref{def-uv-asym},
	the conjugate reduction $v=u^*$ holds under the following constraints:
	\begin{align}\label{conjugate-reduction-FL}
		\bs L=-\bs K^\dagger,~~~~\bs s_1^T=-\mr i\bs r_1^\dagger\bs K^\dagger,~~~~\bs s_2^T=\bs r_2^\dagger,
	\end{align}
	i.e.,
	\begin{align}\label{conjugate-reduction-FL-2}
		l_i=-k_i^*,~~~~\sigma_1(l_i)=-\mr i(\rho_1(k_i))^*k_i^*,~~~~\sigma_2(l_i)=(\rho_2(k_i))^*,~~~~i=1,\cdots,N.
	\end{align}
Here $\bs K^\dag=(\bs K^*)^T$.
\end{theorem}

\begin{proof}
First, the dispersion relations in \eqref{dispersion-SDYM} are consistent
with \eqref{conjugate-reduction-FL}.
Then, by applying the conjugate reduction \eqref{conjugate-reduction-FL-2} to \eqref{M-Cauchy-matrix},
one has $\bs M=\bs M^\dag$ and obtains
	\begin{align}
		\bs K\bs M+\bs M^\dagger\bs K^\dagger=\bs r_2\bs r_2^\dagger.
	\end{align}
Through a direct calculation we have
	\begin{align*}
		U_{21}^\dagger V_{22}&=(\bs s_2^T\bs M^{-1}\bs r_1)^\dagger(1-\bs s_2^T\bs M^{-1}\bs K^{-1}\bs r_2) \\
		&=\bs r_1^\dagger\bs M^{-\dagger}\bs r_2-\bs r_1^\dagger\bs M^{-\dagger}\bs r_2\bs r_2^\dagger\bs M^{-1}\bs K^{-1}\bs r_2 \\
		&=-\bs r_1^\dagger\bs K^\dagger\bs M^{-1}\bs K^{-1}\bs r_2=-\mr i\bs s_1^T\bs M^{-1}\bs K^{-1}\bs r_2=\mr iV_{12},
	\end{align*}
which gives rise to $v=u^*$ in light of the formula \eqref{def-uv-asym}.
The proof is completed.

\end{proof}

Based on the results in Sec.\ref{sec-3-2-1},
explicit $N$-soliton solution of the FL equation can be presented as the following.
\begin{theorem}\label{thm-3}
	Suppose
	\bsb
	\begin{align}
		&\bs K=\mr{diag}(k_1,\cdots,k_N),~~~~\bs r_j=(\rho_j(k_1),\cdots,\rho_j(k_N))^T,~~~~j=1,2, \\
		&\bs M=(M_{ij})_{N\times N},~~~~M_{ij}=\frac{\rho_{2}(k_i)(\rho_{2}(k_j))^*-\mr i\rho_{1}(k_i)(\rho_{1}(k_j))^*k_j^*}{k_i+k_j^*},
	\end{align}
	\esb
	where
	\bsb\label{PWF-asym}
	\begin{align}\label{PWF-asym-1}
		\rho_1(k_i)=\exp\left(\frac{1}{2k_i}x+\frac{k_1}{2}t+\lambda_1(k_i)\right),~~~~
\rho_2(k_i)=\exp\left(-\frac{1}{2k_i}x-\frac{k_i}{2}t+\lambda_2(k_i)\right),
	\end{align}
	or alternatively
	\begin{align}
		\label{PWF-asym-2}&\rho_1(k_i)=\exp\left(\frac{1}{k_i}x+k_it+\lambda_1(k_i)\right),~~~~
\rho_2(k_i)=\exp\left(\lambda_2(k_i)\right), \\
		\label{PWF-asym-3}&\rho_1(k_i)=\exp(\lambda_1(k_i)),~~~~
\rho_2(k_i)=\exp\left(-\frac{1}{k_i}x-k_it+\lambda_2(k_i)\right).
	\end{align}
	\esb
Then the following function
	\begin{align}\label{u-con-re}
		u_{[\mr{KP}]}=\bs r_2^\dagger\bs M^{-1}\bs r_1
	\end{align}
provides a $N$-soliton solution of the FL equation \eqref{FL-equation-2}.
\end{theorem}

\subsection{The AKNS-type solution}\label{sec-4-2}

For the solutions of the pKN$(-1)$ derived in Sec.\ref{sec-3-2-2} from the
Cauchy matrix scheme of the AKNS-type,
their reductions can be describe below.

\begin{theorem}\label{thm-4}
For $(u,v)$ defined in  \eqref{def-uv-sym}, the conjugate reduction $v=u^*$ holds under the following constraints:
	\begin{align}
		\bs K_2=-\bs K_1^\dagger,~~~~\bs s_2=\bs r_1^*,~~~~\bs r_2=\mr i\bs K_1^*\bs s_1^*,
	\end{align}
i.e.,
	\begin{align}\label{4.8}
		l_i=-k_i^*,~~~~\sigma_2(l_i)=(\rho_1(k_i))^*,~~~~\rho_2(l_i)=\mr ik_i^*(\sigma_1(k_i))^*.
	\end{align}
\end{theorem}

\begin{proof}
	It is not difficult to verify the following Sylvester equations holds for the settings in \eqref{def-sym}:
	\begin{align}
		\bs K_2\bs M_2-\bs M_2\bs K_1=\bs r_2\bs s_1^T.
	\end{align}
Then, by applying \eqref{4.8} to \eqref{M1} and \eqref{M2}, one obtains
	\begin{align}
		\bs M_1^\dagger=\bs M_1,~~~~
\bs M_2^\dagger=-\bs K_2^{-1}\bs M_2\bs K_2^\dagger=\bs K_2^{-1}\bs M_2\bs K_1.
	\end{align}
	A direct calculation yields
	\begin{align*}
		u_3^\dagger v_4&=(\bs s_1^T(\bs I_N-\bs M_1\bs M_2)^{-1}\bs r_1)^\dagger
		(1+\bs s_1^T(\bs I_N-\bs M_1\bs M_2)^{-1}\bs M_1\bs K_2^{-1}\bs r_2) \\
		&=\bs r_1^\dagger(\bs I_N-\bs M_2^\dagger\bs M_1^\dagger)^{-1}\bs s_1^*-\mr i\bs r_1^\dagger
(\bs I_N-\bs M_2^\dagger\bs M_1^\dagger)^{-1}\bs K_2^{-1}\bs r_2\bs s_1^T
\bs M_1(\bs I_N-\bs M_2\bs M_1)^{-1}\bs K_2^{-1}\bs r_2 \\
		&=-\mr i\bs r_1^\dagger(\bs I_N-\bs M_2^\dagger\bs M_1^\dagger)^{-1}
(\bs I_N-\bs K_2^{-1}\bs M_2\bs K_1\bs M_1)(\bs I_N-\bs M_2\bs M_1)^{-1}\bs K_2^{-1}\bs r_2 \\
		&=-\mr i\bs s_2^T(\bs I_N-\bs M_2\bs M_1)^{-1}\bs K_2^{-1}\bs r_2=\mr iv_2,
	\end{align*}
	which completes the proof in the light of \eqref{def-uv-sym}.
\end{proof}

Explicit $N$-soliton solution formula of the FL equation can be presented below.

\begin{theorem}\label{thm-5}
The following function
	\begin{align}\label{u-con-re-2}
		u_{[\mr{AKNS}]}=\bs s_1^T(\bs I_N-\bs M_1\bs M_2)^{-1}\bs r_1
	\end{align}
gives a $N$-soliton solution of the FL equation \eqref{FL-equation-2},
where
	\bsb
	\begin{align}
		&\bs K=\mr{diag}(k_1,\cdots,k_N), \\
		&\bs r_1=(\rho_1(k_1),\cdots,\rho_1(k_N))^T,~~~~\bs s_1=(\sigma_1(k_1),\cdots,\sigma_1(k_N))^T,\\
		&\bs M_1=(M_{1,ij})_{N\times N},~~~~M_{1,ij}=\frac{\rho_1(k_i)(\rho_1(k_j))^*}{k_i+k_j^*}, \\
		&\bs M_2=(M_{2,ij})_{N\times N},~~~~M_{2,ij}=-\frac{\mr ik_i^*(\sigma_1(k_i))^*\sigma_1(k_j)}{k_i^*+k_j},
	\end{align}
	\esb
the plane wave factors are give by
	\bsb
	\begin{align}\label{PWF-sym-1}
		\rho_1(k_i)=\exp\left(\frac{1}{2k_i}x+\frac{k_i}{2}t+\lambda_1(k_i)\right),  ~~~~
\sigma_1(k_i)=\exp\left(\frac{1}{2k_i}x+\frac{k_i}{2}t+\mu_1(k_i)\right),
	\end{align}
	or alternatively,
	\begin{align}
		\label{PWF-sym-2}&\rho_1(k_i)=\exp\left(\frac{1}{k_i}x+k_it+\lambda_1(k_i)\right),~~~~
\sigma_1(k_i)=\exp\left(\mu_1(k_i)\right), \\
		\label{PWF-sym-3}&\rho_1(k_i)=\exp(\lambda_1(k_i)),~~~~
\sigma_1(k_i)=\exp\left(\frac{1}{k_i}x+k_it+\mu_1(k_i)\right).
	\end{align}
\esb

\end{theorem}

\subsection{Nonlocal reduction and multiple-pole solution}\label{sec-4-3}

Both the KP-type and AKNS-type solutions for the pKN$(-1)$ system
admit a nonlocal reduction
\begin{equation}\label{uv-nonl}
v(x,t)=-\mr iu^*(-x,-t),
\end{equation}
which gives rise to a nonlocal FL equation from the pKN$(-1)$ system:
\begin{align}
	u_{xt}-u+2u^*(-x,-t)uu_x=0.
\end{align}
The reductions to get solutions of this equation are given by
\begin{itemize}
	\item \textbf{Nonlocal reduction for the KP-type:}
	\begin{align}
		\bs L=\bs K^\dagger,~~~~\bs s_1^T=\bs r_1^\dagger(-x,-t)\bs K^\dagger,~~~~
\bs s_2^T=\bs r_2^\dagger(-x,-t),
	\end{align}
	
	\item \textbf{Nonlocal reduction for the AKNS-type:}
	\begin{align}
		\bs K_2=\bs K_1^\dagger,~~~~\bs s_2=\bs r_1^*(-x,-t),~~~~
\bs r_2=-\bs K_1^*\bs s_1^*(-x,-t).
	\end{align}
\end{itemize}
The proof is similar to Theorem \ref{thm-2} and \ref{thm-4} and we skip it.

The Cauchy matrix approach can be used to derive not only $N$-soliton solutions
by considering $\bs K$ as a diagonal matrix, but also  multiple-pole solution,
in which $\bs K$ and $\bs L$ ($\bs K_1$ and $\bs K_2$ in the AKNS-type)
are assumed to be   Jordan block matrices.
The construction of multiple-pole solution had been fully discussed for the case of the SDYM equation
in our previous work \cite{LSS-2023},
one can refer to it.
As an additional contribution to the completeness of this paper,
we will provide the multiple-pole solution of FL equation in appendix \ref{App-1},
 where we will  only consider the KP-type case, while the construction 
 for the AKNS-type multiple-pole solutions is similar.

\section{Concluding remarks}\label{sec-5}

In this paper we have shown how the unreduced FL system, i.e. the pKN$(-1)$ system \eqref{FL-unreduce},
arose from the reduction of  the general SDYM equation \eqref{Miura-JK}.
The reduction was shown to be realized in the two Cauchy matrix schemes,
namely, the KP-type and the AKNS-type.
Consequently, two types of solutions of the pKN$(-1)$ system were constructed,
which turn out to equivalent under certain reflection transformation of coordinates.
These solutions allow further (conjugate) reductions and at last yield solutions for the FL equation \eqref{FL-equation-2}.

It can be identified that our solutions for the FL equation are the same as those obtained
by Matsuno from bilinear approach \cite{Matsuno-2012}.
One can refer to Appendix \ref{App-2} of this paper for more details.
In addition, in \cite{Ye-2023},  vectorial Darboux transformation via bidifferential graded algebra techniques
was employed to construct solutions for the unreduced FL system.
Their solution (see Theorem 4.2 in \cite{Ye-2023})
also coincides with our KP-type Cauchy matrix solution with $\bs a=\sigma_3$.
However, it seems not easy to identify the relation between the solutions  in Cauchy matrix form of this paper
and in  double Wronskian form obtained in \cite{Liu-2022}.
Besides solitons, based on the Cauchy matrix approach, one can construct more solutions other than
solitons, for example, multiple-pole solutions.
In Appendix \ref{App-1}, some formulae of multiple-pole solutions for the  pKN$(-1)$ system
and the FL equation are presented.

In this paper, apart from the  pKN$(-1)$ system, we also derived the AKNS$(-1)$ system \eqref{AKNS-1},
which naturally appears as the $K$-formulation of the SDYM equation under
a 2-dimensional reduction.
Note that there is a Riccati-type Miura transformation to connect the AKNS$(-1)$ system
and the  pKN$(-1)$ system:
\begin{align}\label{Riccati-Miura}
	K_{12}=uv^2-\mr iv_t,
\end{align}
which was first found in \cite{Ye-2023}
and played a crucial role in \cite{Ye-2023} in constructing solutions for the FL equation.
Actually, \eqref{Riccati-Miura} is nothing but  \eqref{thm-1-process} in the proof of Theorem \ref{thm-1},
which naturally appears from the Miura transformation \eqref{Miura-pKN}.

The paper added an important example to Ward's conjecture on the reduction of the SDYM equation.
It may indicate the possibility of other types of solutions for the FL equation in light of such reductions.
For example, the SDYM equation is famous for instantons \cite{Atiyah-1977}
and the Atiyah-Hitchin-Drinfield-Manin (AHDM) ansatz \cite{Atiyah-1978};
it also admits quasi-Wronskian type solutions represented using quasideterminant \cite{Hamanaka-2020,Huang-2021}.
Whether these solutions could be reformulated in the Cauchy matrix approach and
be reduced to the lower dimensional cases would be  an interesting topic.
In addition, our research conducted in this paper also indicates the Cauchy matrix structure of the
Kaup-Newell hierarchy, which will be investigated separately.
Finally, considering the Cauchy matrix approach is also a powerful tool to implement
integrable discretization, e.g. \cite{CWZ-TMP-2023},
one can consider discretization of some equations in the  Kaup-Newell hierarchy (cf.\cite{NCQ-1983-PLA}).
Compared with the discretization of the AKNS hierarchy,
this is not well understood in literature.

\vskip 20pt
\subsection*{Acknowledgments}

This project is supported by the NSFC grant (Nos. 12271334, 12411540016 and 12201329).

\appendix

\section{Construction of the KP-type multiple-pole solutions }\label{App-1}

The construction for multiple-pole solutions is much more complicated.
We start our construction by introducing the following
$N$th-order lower triangle Toeplitz matrix generated by a function $a(k)$:
\begin{align}\label{def-F}
	\bs F_k^{[N]}[a(k)]=
	\begin{pmatrix}
		a & 0 & 0 & \cdots & 0 \\
		\frac{\partial_ka}{1!} & a & 0 & \cdots & 0 \\
		\frac{\partial_k^2a}{2!} & \frac{\partial_ka}{1!} & a & \cdots & 0 \\
		\vdots & \vdots & \vdots & \ddots & \vdots \\
		\frac{\partial_k^{N-1}a}{(N-1)!} & \frac{\partial_k^{N-2}a}{(N-2)!} &
\frac{\partial_k^{N-3}a}{(N-3)!} & \cdots & a
	\end{pmatrix}_{N\times N}.
\end{align}
Note that the set of all nonsingular $N$-order lower triangle Toeplitz matrices compose
an Abelian group.
Thus we have the following commutative relation
\begin{align*}
	\bs F_k^{[N]}[a(k)]\bs F_l^{[N]}[b(l)]=\bs F_l^{[N]}[b(l)]\bs F_k^{[N]}[a(k)].
\end{align*}
It is notable that by setting $a(k)=k$ in \eqref{def-F},
matrix $\bs F_k^{[N]}[a(k)]$ yields a $N$th-order Jordan block matrix of $k$, which is presented as
\begin{align*}
	\bs F_k^{[N]}[k]=\bs J^{[N]}[k]=
	\begin{pmatrix}
		k & 0 & 0 & \cdots & 0 \\
		1& k& 0 & \cdots & 0 \\
		0& 1& k & \cdots & 0 \\
		\vdots & \vdots & \vdots & \ddots & \vdots \\
		0 & 0& 0 & \cdots & k
	\end{pmatrix}_{N\times N}.
\end{align*}

Then we introduce a lemma as follows.

\begin{lemma}\label{lem-syl}
	For the Sylvester equation
	\begin{align}
		\bs K\bs M-\bs M\bs L=\bs r\bs s^T,
	\end{align}
	where $\bs K,\bs L,\bs r,\bs s^T$ are defined as
	\bsb
	\begin{align}
		&\bs K=\bs J^{[N]}[k],~~~~\bs L^T=\bs J^{[M]}[l], ~~~~
		\bs R=\bs F_k^{[N]}[\rho(k)],~~~~\bs S=\bs F_l^{[M]}[\sigma(l)], \\
		&\bs r=\bs R\mf e_N=\bs F_k^{[N]}[\rho(k)]\mf e_N,~~~~
\bs s=\bs S\mf e_M=\bs F_l^{[M]}[\sigma(l)]\mf e_M,~~~~
		\mathbf e_{N}=(\underbrace{1,0,0,\cdots,0}_{N-\mathrm{dimensional}})^T,
	\end{align}
	\esb
its solution matrix $\bs M$ can be formulated as
	\begin{align}\label{M-expand}
		\bs M=\bs R\bs G\bs S^T=\bs F_k^{[N]}[\rho(k)]\cdot\bs G\cdot\left(\bs F_l^{[M]}[\sigma(l)]\right)^T,
	\end{align}
	where $\bs G$ is a matrix determined by $\bs K$ and $\bs L$:
	\begin{align}\label{G-expand}
		\bs G=(g_{i,j})_{N\times M},~~~~
		g_{i,j}=
		\begin{pmatrix}
			i+j-2 \\
			i-1
		\end{pmatrix}
		\frac{(-1)^{i-1}}{(k-l)^{i+j-1}},
	\end{align}
	where
	\begin{align}
		\begin{pmatrix}
			n \\
			m
		\end{pmatrix}
		=
		\frac{n!}{m!(n-m)!},~~~~
		n\geq m,~~~~
		m,n\in\mb Z^+.
	\end{align}
\end{lemma}

\begin{proof}
	Firstly, with \eqref{M-expand} one can rewrite the Sylvester equation as
	\begin{align}\label{Process-Syl}
		\bs K\bs R\bs G\bs S^T-\bs R\bs G\bs S^T\bs L=\bs R\mf e_N\mf e_M^T\bs S^T,
~~~~\Rightarrow~~~~
		\bs R(\bs K\bs G-\bs G\bs L)\bs S^T=\bs R\mf e_N\mf e_M^T\bs S^T.
	\end{align}
Thus $\bs G$ satisfies the following Sylvester equation
	\begin{align}\label{Syl-G}
		\bs K\bs G-\bs G\bs L=\mf e_N\mf e_M^T,
	\end{align}
which has been solved in \cite{Zhang-2013}, where the solution is given by \eqref{G-expand}.
	
\end{proof}

Then we have the following theorem, which devotes to constructing
the KP-type multiple-pole solution of pKN$(-1)$ system.

\begin{theorem}\label{thm-multi-pole}
	To derive multiple-pole solution via Cauchy matrix approach, we  set
	\bsb\label{KL-multi-pole}
	\begin{align}
		&\bs K=\bs J^{[N]}[k],~~~~\bs L^T=\bs J^{[N]}[l],~~~~
		\bs R_j=\bs F_k^{[N]}[\rho_j(k)],~~~~\bs S_j=\bs F_l^{[N]}[\sigma_j(l)],~~~~j=1,2, \\
		&\bs r_j=\bs R_j\mf e_N=\bs F_k^{[N]}[\rho_j(k)]\mf e_N,~~~~
\bs s_j=\bs S_j\mf e_N=\bs F_l^{[N]}[\sigma_j(l)]\mf e_N,
	\end{align}
where $\rho_j(k)$ and $\sigma_j(l)$ are defined as in \eqref{def-RS}.
	\esb
	Then $\bs M$ can be constructed as:
	\begin{align}\label{M-multipole}
		\bs M=\bs R_1\bs G\bs S_1^T+\bs R_2\bs G\bs S_2^T,
	\end{align}
	where $\bs G$ follows the expression \eqref{G-expand}.
	By definition \eqref{def-uv-asym} and \eqref{KL-multi-pole},
$(u,v)$ solves the pKN$(-1)$ system.
\end{theorem}

\begin{remark}\label{rem-KM+ML}
	For the Sylvester equation
	\begin{align*}
		\bs K\bs G+\bs G\bs L=\mf e_N\mf e_M^T,
	\end{align*}
	where $\bs K=\bs J^{[N]}[k]$ and $\bs L^T=\bs J^{[M]}[l]$, solution $\bs G$ is constructed as
	\begin{align}\label{Rem-G}
		\bs G=(g_{i,j})_{N\times M},~~~~
		g_{i,j}=
		\begin{pmatrix}
			i+j-2 \\
			i-1
		\end{pmatrix}
		\frac{(-1)^{i+j}}{(k+l)^{i+j-1}}.
	\end{align}
	If one replace $\bs G$ in Theorem \ref{thm-multi-pole} with \eqref{Rem-G},
then $\bs M$ solves the Sylvester equation
	\begin{align*}
		\bs K\bs M+\bs M\bs L=\bs r\bs s^T.
	\end{align*}
\end{remark}

As for the explicit formula of multiple-pole solution of the FL equation, we have the following theorem.
\begin{theorem}\label{thm-mullti-pole-FL}
	To derive multiple-pole solution of the FL equation, we  set
	\begin{align}
		\bs K=\bs J^{[N]}[k],~~~~\bs R_j=\bs F_k^{[N]}[\rho_j(k)],~~~~\bs r_j=\bs R_j\mf e_N,~~~~j=1,2,
	\end{align}
where $\rho_j(k)$ is defined as in \eqref{PWF-asym-1}.
	The reduced Cauchy matrix $\bs M$ is constructed by
	\begin{align}\label{M-reduced}
		\bs M=\bs R_1\bs G\bs R_1^\dagger-\mr i\bs R_2\bs G\bs R_2^\dagger\bs K^\dagger,
	\end{align}
	where
	\begin{align}\label{G-reduce}
		\bs G=(g_{i,j})_{N\times N},~~~~
		g_{i,j}=
		\begin{pmatrix}
			i+j-2 \\
			i-1
		\end{pmatrix}
		\frac{(-1)^{i+j}}{k+k^*}.
	\end{align}
Then \eqref{u-con-re} provides multiple-pole solution for the FL equation \eqref{FL-equation-2}.
\end{theorem}

\begin{remark}
	Notice that we use the expression \eqref{Rem-G} with $l=k^*$ instead of \eqref{G-expand} with $l=-k^*$ in this theorem. In fact, if we started from
	\begin{align}
		\bs K\bs M-\bs M\bs L=\bs r\bs s^T,~~~~\bs L=-\bs K^\dagger,~~~~\bs K=\bs J^{[N]}[k],
	\end{align}
	we could not use Theorem \ref{thm-multi-pole}. In this case, matrix $\bs L$ will be of the form
	\begin{align}
		\bs L^T=
		\begin{pmatrix}
			-k^* & 0 & 0 & \cdots & 0 \\
			-1 & -k^* & 0 & \cdots & 0 \\
			0& -1 & -k^* & \cdots & 0 \\
			\vdots & \vdots & \vdots & \ddots & \vdots \\
			0 & 0& 0 & \cdots & -k^*
		\end{pmatrix}_{N\times N},
	\end{align}
	which is not the standard Jordan matrix in our definition.
	To tackle the problem, we consider $\bs L'=-\bs L$ and the following set:
	\begin{align}
		\bs K\bs M+\bs M\bs L'=\bs r\bs s^T,~~~~\bs L'=\bs K^\dagger,~~~~\bs K=\bs J^{[N]}[k],
	\end{align}
	which indicates $(\bs L')^T=\bs J^{[N]}[k^*]$.
Thus one can use the results in Remark \ref{rem-KM+ML} to construct the multiple-pole solution.
\end{remark}

\section{Correspondence to the Matsuno solution}\label{App-2}

In the literature \cite{Matsuno-2012}, for the FL equation (to avoid misunderstandings,
we use $\mr u$  rather than $u$ in the following FL equation):
\begin{align}\label{FL-equation-v}
	\mr{u}_{xt}-\mr{u}+2\mr i|\mr{u}|^2\mr{u}_x=0,
\end{align}
the author (Y. Matsuno) introduced a bilinear transformation $\mr{u}=g/f$ and rewrote \eqref{FL-equation-v} as
\bsb
\begin{align}
	&D_xD_t g\cdot f=gf, \\
	&D_t f\cdot f^*=\mr igg^*, \\
	&D_xD_t f\cdot f^*=\mr iD_xg\cdot g^*,
\end{align}
where $D$ is the Hirota's bilinear operator defined as \cite{Hir-1974}
\begin{equation}
D^m_xD^n_t f(x,t)\cdot g(x,t)=(\partial_x-\partial_{x'})^m (\partial_t-\partial_{t'})^n f(x,t)g(x',t')|_{x'=x,t'=t}.
\label{D}
\end{equation}
\esb
Then, by Theorem 3.1 of \cite{Matsuno-2012},
the following constructions give rise to the following bright $N$-soliton solution of the FL equation:
\bsb
\begin{align}
	&f=|W|,~~~~
	g=
	\begin{vmatrix}
		W & \mf z_t^T \\
		\bs 1 & 0
	\end{vmatrix},~~~~
	W=(d_{ij})_{N\times N},~~~~
	d_{ij}=\frac{z_iz_j^*-\mr ip_j^*}{p_i+p_j*}, \\
	&z_i=\exp\left(p_ix+\frac{1}{p_i}t+\zeta_{i0}\right),~~~~
	\mf z_t=(z_1/p_1,\cdots,z_N/p_N),~~~~
	\bs 1=(1,1,\cdots,1).
\end{align}
\esb

Now we consider  transformations:
\begin{align}
	p_i \rightarrow -\frac{1}{k_i},~~~~
	z_i \rightarrow \exp\left(-\frac{1}{k_i}x-k_it+\zeta_{i0}\right),~~~~
	\mf z_t\rightarrow(-k_1z_1,\cdots,-k_Nz_N),
\end{align}
which lead to
\begin{align}
	d_{ij}=-\frac{k_iz_iz_j^*k_j^*+\mr ik_i}{k_i+k_j^*},~~~~
	d_{ji}^*=-\frac{k_iz_iz_j^*k_j^*-\mr ik_j^*}{k_i+k_j^*}.
\end{align}
Let $\rho_2(k_i)=k_iz_i$ and $\rho_1(k_i)=1$ in \eqref{PWF-asym-3}.
Then we have $W=-\bs M^\dagger$, $\bs 1=\bs r_1^\dagger$ and $\mf z_t^T=-\bs r_2$ in Theorem \ref{thm-4}.
Thus we can rewrite Matsuno's solution as
\begin{align}
	\mr u=\frac{
		\begin{vmatrix}
			W & \mf z_t^T \\
			\bs 1 & 0
		\end{vmatrix}
	}{|W|}=
	-\bs 1W^{-1}\mf z_t^T=-\bs r_1^\dagger\bs M^{-\dagger}\bs r_2=-(\bs r_2^\dagger\bs M^{-1}\bs r_1)^\dagger=-u_{[\mr{KP}]}^*,
\end{align}
where $u_{[\mr{KP}]}$ follows the KP-type construction in \eqref{u-con-re} with \eqref{PWF-asym-3}.

Matsuno also indicated in \cite{Matsuno-2012} that there is an alternative expression $\mr{u}=g'/f'$:
\begin{align}
	f'=
	\begin{vmatrix}
		A & I \\
		-I & B
	\end{vmatrix},~~~~
	g'=
	\begin{vmatrix}
		A & I & \mf y_t^T \\
		-I & B & \bs 0^T \\
		\bs 0 & \bs 1 & 0
	\end{vmatrix},
\end{align}
where
\bsb
\begin{align}
	&A=(a_{ij})_{N\times N},~~~~a_{ij}=\frac{y_iy_j^*}{q_i+q_j^*},~~~~y_i=\exp\left(q_ix+\frac{1}{q_i}t+\eta_{i0}\right), \\
	&B=(b_{ij})_{N\times N},~~~~b_{ij}=\frac{\mr iq_j}{q_i^*+q_j},~~~~\mf y_t=(y_1/q_1,\cdots,y_N/q_N).
\end{align}
\esb
Similarly, we take the following transformations:
\begin{align}
	q_i\rightarrow\frac{1}{k_i},~~~~y_i\rightarrow\exp\left(\frac{1}{k_i}x+k_it+\eta_{i0}\right),~~~~
	\mf y_t=(k_1y_1,\cdots,k_Ny_N),
\end{align}
which lead to
\begin{align}
	a_{ij}=\frac{k_iy_iy_j^*k_j^*}{k_i+k_j^*},~~~~
	b_{ij}=\frac{\mr ik_i^*}{k_i^*+k_j}.
\end{align}
Let $\rho_1(k_i)=k_iy_i$ and $\sigma_1(k_i)=1$ in \eqref{PWF-sym-2},
which implies $A=\bs M_1$, $B=-\bs M_2$, $\mf y_t=\bs r_1^T$, $\bs 1=\bs s_1^T$
in Theorem \ref{thm-5}. Then we have
\begin{align}
	&f'=|AB+I|=|\bs I_N-\bs M_1\bs M_2|, \\
	&g'=
	\begin{vmatrix}
		I+AB & \mf y_t^T \\
		\bs 1 & 0
	\end{vmatrix}=
	\begin{vmatrix}
		\bs I_N-\bs M_1\bs M_2 & \bs r_1 \\
		\bs s_1^T & 0
	\end{vmatrix},
\end{align}
which shows
\begin{align}
	\mr u=\frac{g'}{f'}=-\bs s_1^T(\bs I_N-\bs M_1\bs M_2)^{-1}\bs r_1=-u_{[\mr{AKNS}]},
\end{align}
where $u_{[\mr{AKNS}]}$ follows the AKNS-type construction \eqref{u-con-re-2} with \eqref{PWF-sym-2}.

In Proposition 3.1 and 3.2 of \cite{Matsuno-2012}, the equivalence of $(f,g)$ and $(f',g')$ is established,
which coincides with our discovery in section \ref{sec-3-2-3}.
Note that we have proved a more general case for pKN$(-1)$ system.

\small{
	
}

\end{document}